\documentclass[journal]{IEEEtran}

\ifCLASSINFOpdf
\else
\fi
\usepackage{amsfonts}
\usepackage{graphicx}
\usepackage{tabularx} 
\usepackage[english]{babel}

\usepackage{booktabs}   
\usepackage{siunitx}    
\usepackage{adjustbox}   
\usepackage{booktabs,tabularx}
\usepackage{adjustbox}
\usepackage{soul}
\usepackage[minimal,theorems]{textool}
\usepackage{cite}
\usepackage{xurl}        
\usepackage{soul}
\usepackage{xcolor}
\sethlcolor{red}

\soulregister\cite7
\soulregister\ref7
\soulregister\eqref7
\soulregister\footnote7
\usepackage{subcaption}
\usepackage{placeins}
\usepackage{tabularx,booktabs,array}
\usepackage{amssymb}
\usepackage{hyperref}
\usepackage{placeins}
\newtheorem{lemma}{Lemma}
\newtheorem{theorem}{Theorem}
\newtheorem{corollary}{Corollary}
\begin{document}
\bstctlcite{IEEEexample:BSTcontrol}
\setlength{\abovecaptionskip}{2pt}
\setlength{\belowcaptionskip}{0pt}
\title{Human-on-the-loop Distributed Control of Grid-Interactive Buildings for Demand Response Participation}
%
%

\author{Kasra~Mazarei~Saadabadi,~\IEEEmembership{Student Member, IEEE},
Dongming~Wang,~\IEEEmembership{Student Member, IEEE},
Wei~Ren,~\IEEEmembership{Fellow, IEEE},
Alfredo~Martinez-Morales,~\IEEEmembership{Member, IEEE},
and~Hamidreza~Nazaripouya,~\IEEEmembership{Senior Member, IEEE}%
\thanks{K. Mazarei Saadabadi, D. Wang, W. Ren, and A. Martinez-Morales are with the Department of Electrical Engineering, University of California Riverside, Riverside, CA 92521 USA.}
\thanks{H. Nazaripouya is with the School of Electrical and Computer Engineering, Oklahoma State University, Stillwater, OK 74078 USA.}
\thanks{Corresponding author: H. Nazaripouya (e-mail: hanazar@okstate.edu).}}
%



\setcounter{secnumdepth}{4}

\maketitle
\begin{abstract}
\sloppy
This paper proposes a human-on-the-loop distributed consensus control approach for demand-side management across multiple buildings. Specifically, a novel framework is introduced in which a human acts as the non-autonomous leader in consensus control of cooperative buildings participating in demand response programs. In this system, the facility manager in the facility building serves as the leader, determining the participation level of cooperative buildings in demand response events while simultaneously considering occupants' comfort. Cooperative buildings align their responses with the facility management building, despite lacking direct access to the facility manager’s decisions, and this raises two coupled challenges: The manager's actions act as an unknown input to the leader, and the consensus protocol requires a state that cannot be measured. Because the input is unknown, the necessary observer should be an unknown-input observer to guarantee the performance, and the observer-matching condition required by conventional linear designs does not hold for this system. A human-on-the-loop leader-follower consensus protocol is introduced to manage energy use and balance thermal comfort, together with a nonlinear unknown-input sliding-mode observer that reconstructs the local consensus error in finite time without requiring the mentioned condition. Under exact measurements, continuous-time dynamics, and an admissible initial state, a componentwise certificate guarantees bounded tracking and bounds on ramp rate, comfort, and power when the leader power and the ambient temperature remain within prescribed limits. Simulation results validate the approach, demonstrating its ability to achieve consensus, maintain system performance, and enhance the adaptability of power grid operations under various demand response scenarios.\end{abstract}
\vspace{-0.05in}

\IEEEpeerreviewmaketitle
\vspace{-0.1in}
\section{Introduction}
\vspace{-0.05in}
\sethlcolor{blue}
\colorlet{blue}{black}
\IEEEPARstart{P}{ower} systems are increasingly recognized as cyber-physical systems due to their integration of  both digital and physical components. However, beyond these cyber and physical aspects,  humans play a crucial role in the operation and success of these systems, primarily as operators and end users. With the advent of smart grids, which facilitate a more interactive environment between the system and stakeholders, humans have transitioned from passive observers to empowered, active participants. For example, in smart grids, end-users can actively manage their  energy generation and consumption, thereby becoming  prosumers, entities that both produce and consume energy \cite{annaswamy2021cyber}. Within a smart grid framework, end-users have access to various energy management tools, such as demand response (DR) and load-shifting strategies, enabling them to participate in real-time demand-side management \cite{6099519}. 

Numerous studies have examined DR as a key enabler of interactive energy management in modern power systems \cite{6525433}. Prior work has studied DR for peak-load reduction, grid support, and energy-efficiency improvement \cite{SIANO2014461, 5454394, 5930335, 5607339}. While much of the early research has focused on individual-level DR participation \cite{5540263}, recent studies have highlighted the advantages of aggregating loads across multiple agents for more efficient coordination and enhanced flexibility \cite{Hui2022, Zhao2024}. This growing body of work sets the foundation for more advanced DR frameworks, such as those involving coordinated control among multiple buildings.

The concept of grid-interactive and cooperative buildings is an emerging topic in the field of DR \cite{Hurtado2018}. Research indicates that a network of grid-interactive buildings can collaborate to provide a larger aggregated flexible demand to the power system \cite{PINTO2022118497}. In this structure, interconnected buildings coordinate demand adjustments under the supervision of a facility management team, which also serves as the primary interface with the utility for receiving DR signals. The facility manager then decides whether to comply with the DR command, with other buildings following the facility manager's directives.
To minimize the cost of communication infrastructure, it is recommended that each building limit its interactions to nearby neighboring buildings. 

As seen, DR represents a prominent power system application in which human users play a pivotal role.  However,  to date, human involvement in DR has primarily been represented either through its impact on load patterns or as constraints within optimization models \cite{Shafie2018, Tao2022}. Existing DR coordination schemes, however, assume either that the coordinating command is broadcast to all participants or that the leader is autonomous with no exogenous input. Neither case captures human supervision, in which a command drives the fleet but is not communicated to it. 

To this end, this paper proposes a multi-agent control scheme that integrates the facility manager as a supervisory command authority within the control architecture for coordinated management of multiple buildings. Integrating humans into a control system introduces complexities due to variability and nonlinear, unpredictable responses, posing challenges in achieving reliable control. Addressing these complexities necessitates a rigorous  approach to design robust and reliable human-on-the-loop (HOTL) control systems.

Since in this application the proposed HOTL multi-agent  controller aims to coordinate the behavior of a multi-building  system based on human command, it is structured as a leader-following consensus control. The leader-following consensus control (see e.g., \cite{1333204} and \cite{Olfati-Saber2007Consensus}) drives a group of agents toward agreement on designated state variables, thereby achieving coordinated behavior driven by the leader’s prescribed objective \cite{Cao2013}. This approach is particularly well-suited for cyber physical human systems (CPHS) where humans interact with multiple agents in a leader role to coordinate system-wide behavior.  In this structure, humans serve as non-autonomous leaders \cite{10083099} who provide inputs to the system, ensuring that the multi-agent system remains aligned with desired objectives. This enables the design of adaptive and efficient systems capable of achieving consensus between the human leader and multi-agent systems (MASs) \cite{HONG20061177}.

Despite its potential, developing HOTL consensus control \cite{8618994, 10054599, 9310661, 10185611} in DR applications raises two coupled problems: First, the facility manager’s decisions constitute  an unknown input to the leader that reflects complex trade-offs among thermal comfort, grid-service provision, and participation costs. These decisions are neither directly measurable nor predictable by the follower buildings, and they are not communicated to them.

Second, the consensus protocol requires a state that is not directly measurable in practice. In particular, the derivative of the thermal comfort state cannot be reliably obtained by numerical differentiation of noisy sensor measurements. An observer is therefore required, and because the leader's input is unknown, it must estimate this state without access to that input. To address this requirement, two primary approaches can be considered: Full-order distributed observers, which provide complete state estimates while rejecting known disturbances \cite{9062580}, and interval observers that bound estimation uncertainty when dealing with human inputs \cite{10054599}. 

However, the structural prerequisites for canceling or compensating unknown inputs in conventional linear observers, referred to as observer matching conditions, are not satisfied by the thermal model of the buildings. As a result, existing linear observer design methods such as those proposed in \cite{9062580, 280770, 1278} are not applicable in this context. This second problem is a direct consequence of the first: because the manager's command is a human decision, it cannot be known, so the observer must estimate the state without it, and this is what requires the matching condition in conventional linear designs.
For systems with unknown inputs, sliding-mode observers are often used when the typical conditions for directly estimating the system's states are not met \cite{1104530, tranninger2023unknown}. However, these approaches are typically focused on single-agent systems.

This paper addresses both problems with a HOTL, distributed multi-agent consensus control approach for cooperative heterogeneous buildings that communicate only with their neighbors. 

The contributions of this paper are as follows:
\begin{enumerate}
    \item Formulating multi-building DR coordination as a leader–follower consensus problem in which the facility manager’s curtailment command is an unknown input to the followers, and designing a saturated consensus protocol and a leader-side governor with certified bounds on tracking error, ramp rate, comfort, and power.
    \item A distributed super-twisting sliding-mode observer that, in closed loop with the saturated consensus protocol, reconstructs the consensus error in finite time from relative outputs and each building's own input, without requiring the observer-matching condition.
\end{enumerate}
\vspace{-0.15pt}
The remainder of the article is organized as follows. In Section \ref{seq:2}, the model of the system and its description are presented. The HOTL leader-following consensus controller, including the sliding-mode observer,  is detailed in Section \ref{seq:3}. The proof of the stability of the proposed method is presented in Section \ref{seq:4}. In Section \ref{seq:5}, simulation results are provided to demonstrate the effectiveness of the proposed method. Finally, conclusion remarks are given in Section \ref{seq:6}.
\vspace{-0.3pt}
\section{Problem Statement and System Modeling}\label{seq:2}
\vspace{-0.05in}
\subsection{Problem Statement}
\vspace{-0.05in}
A multi-building campus environment is considered in this paper, where buildings communicate locally with one another to participate in DR events. Within this system, a designated facility building receives the DR signal from the utility and serves as the leader.

When a utility issues a DR signal, the facility manager determines the acceptable comfort compromise. The facility manager then plans a power-curtailment trajectory for the facility building. The remaining heterogeneous buildings track the facility building’s state and adjust their power demand accordingly, without access to the facility manager's decisions.

Given this framework, the system is modeled as a leader-follower MAS, where the cooperative buildings aim to track the state of the facility building. Additionally, to reduce communication overhead, each building exchanges information solely with its immediate neighbors. 

To ensure both thermal comfort and effective load reduction during DR events, buildings must reach consensus on the leader's state. Unlike autonomous leaders, the human operator provides non-autonomous and unknown control inputs, which are unmeasurable and unpredictable. This issue challenges the design of traditional consensus mechanisms. In addition, the consensus protocol requires a state that must be estimated, and because the leader's input is unknown, this estimation with a conventional linear observer requires the observer-matching condition, which is violated here as the unknown input does not appear in the first derivative of the measured output. 

\vspace{-0.2in}
\subsection{Basic Graph Theory}\label{sec:basic_graph}
\vspace{-0.05in}
Consider a MAS composed of $N$ follower agents and one leader. The interaction between agents and the leader is modeled as a directed graph $\mathcal{G} = (\mathcal{V}, \mathcal{E})$, where $\mathcal{V} = \{v_0, v_1, \dots, v_N\}$ is the set of nodes, and $\mathcal{E} \subseteq \mathcal{V} \times \mathcal{V}$ represents the set of edges, with $(v_j, v_i) \in \mathcal{E}$ indicating that agent $i$ receives information from agent $j$. The leader is represented by node $v_0$. The follower adjacency matrix of the graph is denoted as {\color{blue}$\mathbf{A} = [\mathbf{a_{ij}}] \in \mathbb{R}^{N \times N}$}, where {\color{blue}$\mathbf{a_{ij}} > 0$} if $(v_j, v_i) \in \mathcal{E}$ and {\color{blue}$\mathbf{a_{ij}} = 0$} otherwise. The in-degree matrix {\color{blue}$\mathbf{D}$} is defined as {\color{blue}$\textbf{D} = \text{diag}(\mathbf{d_1}, \mathbf{d_2}, \dots, \mathbf{d_N})$}, where {\color{blue}$\mathbf{d_i} = \sum_{j=1}^{N} \mathbf{a_{ij}}$} represents the in-degree of node $i$. The Laplacian matrix is then given by {\color{blue}$\mathbf{L} = \mathbf{D} - \mathbf{A}$}. The leader-pinning weights are $a_{i0} >0$ if $(v_{0}, v_i) \in \mathcal{E}$, otherwise $a_{i0} = 0$. And $D_0 = \text{diag}(a_{10},\dots,a_{N0}). $

To ensure consensus within the system, the communication graph $\mathcal{G}$  should have a spanning tree, with the leader $v_0$ acting as the root. The consensus protocol is designed such that all agents synchronize with the leader's state using only local information from their neighbors. Furthermore, the leader's control input $u_0(t)$ is unknown to followers, and only the leader's comfort output, not its comfort rate, is available to the pinned followers ($a_{i0}>0$).
\vspace{-0.15in}
\subsection{System Model}
\vspace{-0.05in}
In this section, the thermal dynamics of a building equipped with multiple Thermostatically Controlled Loads (TCLs), such as air conditioners or heaters, are discussed. Each building aims to maintain a comfortable indoor temperature while minimizing energy consumption. Modeling  begins with  the thermal behavior of a single building and is then extended to a multi-building scenario.

The average indoor temperature of the $i^{th}$ building, which contains $n_i$ TCLs, can be described by the  differential equation in (\ref{eq:1}) \cite{kundu2011modelingcontrolthermostaticallycontrolled, wang2019distributed}:
\vspace{-0.1in}
{\color{blue}
\begin{equation}\label{eq:1}
C_{i}^{th}\frac{d\theta_i(t)}{dt}=\frac{\theta^{amb}(t) - \theta_i(t)}{R_i^{th}}-\eta_i\frac{P_i(t)}{n_i}
\end{equation}
}
where $\theta_i(t)$ represents the indoor temperature of the $i^{th}$ building in degrees Celsius ($^{\circ}C$). $\theta^{amb}(t)$ is the ambient (outdoor) temperature. $C_i^{th}$ is the thermal capacitance of the $i^{th}$ building, indicating its ability to store thermal energy. $R_i^{th}$ is the thermal resistance of the $i^{th}$ building, indicating to what extent the building resists heat flow. $\eta_i$ is the thermal coefficient, which is positive for cooling devices (e.g., air conditioners) and negative for heating devices (e.g., heaters). $P_i(t) = \sum_{j=1}^{n_i}m^j(t)P_i^j(t)$ is the total power demand of all TCLs in the $i^{th}$ building, where $m^j(t)$ is a binary variable indicating whether the $j^{th}$ TCL is on $(1)$ or off $(0)$, and $P_i^j(t)$ is the power demand of the $j^{th}$ TCL in the $i^{th}$ building. 

Given that different buildings have different preferred indoor temperature ranges, a thermal comfort index $\varepsilon_i(t)$ is defined by normalizing the indoor temperature relative to the desired  temperature range, as shown in (\ref{eq:2}). This index helps quantify and control the thermal comfort of each building.

\begin{equation}\label{eq:2}
\varepsilon_i(t)=\frac{\theta_i(t)-\theta^{set}_i+\Delta\theta_i}{2\Delta\theta_i}
\end{equation}
where $\theta^{set}_i$ is the desired setpoint temperature and $\Delta\theta_i$ is the allowable temperature deviation around the desired setpoint. When $\theta_i(t)=\theta_i^{set}$, then $\varepsilon_i(t) = 0.5$ corresponds to the desired thermal-comfort state. Substituting \eqref{eq:2} into \eqref{eq:1}, the dynamics of the thermal comfort index can be expressed as (\ref{eq:3}):

\begin{equation}
\begin{split}\label{eq:3}
\dot\varepsilon_i(t)=-\frac{1}{C_i^{th}R_i^{th}}\varepsilon_i(t)-\frac{\eta_i}{2\Delta\theta_iC_i^{th}n_i}P_i(t)+\\
{\color{blue}\frac{\theta^{amb}(t)-\theta^{set}_i+\Delta\theta_i}{2\Delta\theta_iC^{th}_iR^{th}_i}}
\end{split}
\end{equation}
where $n_i$ is the number of TCLs in the $i^{th}$ building. Defining the auxiliary state $\xi_i(t)=\dot\varepsilon_i(t)$ allows the comfort-index dynamics to be rewritten as the second-order state-space model in (\ref{eq:model4}):
\begin{equation}\label{eq:model4}
\xi_i(t)= a_i\varepsilon_i(t)+b_iP_i(t)+{\color{blue}r_i(t)}
\end{equation}
where, $a_i=-\frac{1}{C_i^{th}R_i^{th}}$, $b_i=-\frac{\eta_i}{2\Delta\theta_iC_i^{th}n_i}$, and {\color{blue}$r_i=\frac{\theta^{amb}(t)-\theta^{set}_i+\Delta\theta_i}{2\Delta\theta_iC^{th}_iR^{th}_i}$}.

For a reference ambient temperature $T^\star$, the buildings are at equilibrium at a common comfort level $\varepsilon^\star$, with $\xi_i=0$, when each building draws
\begin{equation}
P_i^\star = -\frac{a_i\varepsilon^\star + r_i(T^\star)}{b_i},
\label{eq:Pstar}
\end{equation}
where $r_i(T^\star)$ is $r_i$ evaluated at $\theta_{\mathrm{amb}}=T^\star$. The power $P_i^\star$ is fixed at $T^\star$ and is not adjusted with the weather. The deviations of the leader power and the ambient temperature from this equilibrium are denoted by $p=P_0-P_0^\star$ and $\vartheta=\theta_{\mathrm{amb}}-T^\star$.

Thus, the second-order state-space equation of the multi-agent system can be written as (\ref{eq:5}):

\begin{equation}\label{eq:5}
\left\{\begin{matrix}
\dot\varepsilon_i(t)=\xi_i(t) \\
\dot\xi_i(t)=a_i\xi_i(t)+b_iu_i(t)+ {\color{blue}w_i}
\end{matrix}\right.
\end{equation}
where$\ u_i(t)= \dot P_i(t)$ is the control input for trajectory tracking, and {\color{blue}$w_i (t) = \dot r_i (t)$} represents the disturbance caused by the change of the ambient temperature, with $w_i(t)=\ell_i\dot{\theta}_{\mathrm{amb}}(t)$ and $\ell_i=(2\Delta\theta_i C^{\mathrm{th}}_i R^{\mathrm{th}}_i)^{-1}$. The initial states satisfy \eqref{eq:model4}.

The ambient-temperature term $w_i=\ell_i\dot{\vartheta}$, with $\dot{\vartheta}=\dot{\theta}_{\mathrm{amb}}$, is retained, and the state-space equation takes the form expressed in \eqref{eq:ss}:
\vspace{-0.05in}
\begin{equation}\label{eq:ss}
\dot{x}_i = A_i x_i + B_i u_i +\begin{bmatrix}0\\ \ell_i\end{bmatrix}\dot{\vartheta}, \quad y_i = C x_i
\end{equation}
where $A_i = \begin{bmatrix} 0 & 1\\ 0 & a_i \end{bmatrix}$, $B_i = \begin{bmatrix} 0\\ b_i \end{bmatrix}$, $C = \begin{bmatrix} 1 & 0 \end{bmatrix}$ with agent-specific thermal parameters $a_i$ and $b_i$.

The comfort and power guarantees hold while $\vartheta$ and $\dot{\vartheta}$ stay within the limits given in Section~\ref{sec:control_objective}.

\vspace{-0.17in}
\section{Proposed Control Framework}\label{seq:3}
\vspace{-0.05in}
\subsection{Control objective}\label{sec:control_objective}
\vspace{-0.05in}
The primary objective of the control system is to coordinate the participation of multiple buildings in DR events while maintaining indoor temperatures within an acceptable range and ensuring consistent comfort levels across all participating buildings during normal operation. Here, the acceptable range means $\varepsilon_i(t) \in [0,1]$, i.e., $\theta_i \in [\theta_i^{set}-\Delta\theta_i, \theta_i^{set}+\Delta\theta_i]$. The facility manager specifies the curtailment target, and a leader-side governor keeps the leader power within an admissible band and ramp-rate limit. Under these limits and the ambient-temperature limits, Section \ref{seq:4} shows that the comfort indices of all buildings remain within the acceptable range and that all powers and ramp rates remain within their limits.

During DR events, the utility sends a command signal $P_{DR}$ requesting a reduction in the aggregated power demand of the entire building complex $\Delta P_{tot}$ by a specified amount. The human operator, serving as the facility manager, evaluates the command and determines participation based on buildings' operating conditions and the relative priority assigned to occupants’ comfort. Fig.~\ref{fig:sch} depicts the structure of multiple buildings participating in the DR program.

The facility manager resolves the trade-off between curtailment and occupant comfort using operational and institutional information that is not available to the control system. The command is therefore treated as an unknown exogenous input rather than as the output of a modeled decision rule. The leader power requested by the facility manager is applied through a governor. At intervals of length $h_g$, the governor updates the leader power as
\begin{equation}
P_0[k+1]=P_0[k]+\mathrm{sat}_{s_0h_g}\big(P_0^{\mathrm{r}}[k]-P_0[k]\big),
\label{eq:governor}
\end{equation}
where $P_0^{\mathrm{r}}[k]$ is the request limited to $[P_0^\star-\bar{p},\,P_0^\star]$. Consecutive values are joined linearly, and $u_0=\dot{P}_0$. Starting from $p(0)\in[-\bar{p},0]$, the governor keeps $-\bar{p}\le p(t)\le 0$ and $|\dot{p}(t)|\le s_0$. The ambient temperature is assumed to satisfy $|\vartheta(t)|\le\Theta$ and $|\dot{\vartheta}(t)|\le W$.

\begin{figure}[!t] 
    \centering
    \includegraphics[width=0.75\linewidth]{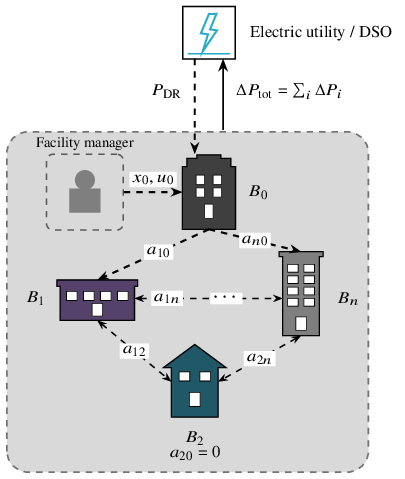} 
    \caption{An overview of the proposed method for the HOTL DR participation and comfort management}
    \vspace{-8pt}
    \label{fig:sch}
    
\end{figure}
\vspace{-0.18in}
\subsection{Consensus Protocol}
\vspace{-0.05in}
To achieve coordinated control across buildings, a consensus protocol is employed so that buildings coordinate their thermal states while balancing comfort and demand reduction. Since each building has physical input limits, the nominal control input is defined using the saturation function \eqref{eq:sat_control} \cite{6237497}:
\begin{equation}
    u_i=\operatorname{sat}_{\bar u_i}(v_i),
    \qquad
    \operatorname{sat}_{\bar u_i}(v_i)
    :=
    \operatorname{sgn}(v_i)\min\{|v_i|,\bar u_i\},
    \label{eq:sat_control}
\end{equation}

Here, $\bar u_i>0$ is the maximum admissible input magnitude. The pre-saturation
control signal \eqref{eq:pre_sat_control} is
\begin{equation}
    v_i=-K_i\delta_i,
    \label{eq:pre_sat_control}
\end{equation}
where $K_i\in\mathbb{R}^{1\times 2}$ is the feedback gain, and the consensus error $\delta_i \in \mathbb{R}^2$ for agent $i$ is defined in (\ref{eq:delta}):
\begin{equation} \label{eq:delta}
\delta_i = \sum_{j \in \mathcal{N}_i} a_{ij}(x_i - x_j) + a_{i0}(x_i - x_0)
\end{equation}
where $\mathcal{N}_i$ denotes the neighbor set of agent $i$, and $a_{i0} > 0$ only if agent $i$ has direct communication with the leader.
The state vector $x_i$ is defined in \eqref{eq:statvec} as:
\begin{equation}
x_i=
\begin{bmatrix}
    \varepsilon_i\\
    \xi_i
\end{bmatrix}
\label{eq:statvec}
\end{equation}

The weights $a_{ij}$ and $a_{i0}$ in \eqref{eq:delta} are the adjacency and pinning
entries defined in Section \ref{sec:basic_graph}. Only $\delta_{i,1}$ is measured, from relative comfort outputs, and $\delta_{i,2}$ is reconstructed by the observer. Stacking \eqref{eq:delta} gives
$\delta=((L+D_0)\otimes I_2)\tilde{x}$ with
$\tilde{x}=\mathrm{col}(x_1-x_0,\dots,x_N-x_0)$; since $L+D_0$ is
nonsingular under Assumption \ref{ass:A1}, $\delta=0$ if and only if
$\tilde{x}=0$. Driving the local errors to zero is therefore equivalent
to leader tracking by every follower, including those with $a_{i0}=0$,
which receive the leader's state transitively along the spanning tree. The analysis assumes exact measured comfort outputs.

Fig. \ref{fig:diag} shows the resulting control structure. Each observer uses the relative comfort outputs and the input of its own building to estimate the consensus error, and the consensus protocol uses this estimate to drive the building toward the leader.

\begin{figure}[!t] 
    \centering\includegraphics[width=0.9\linewidth]{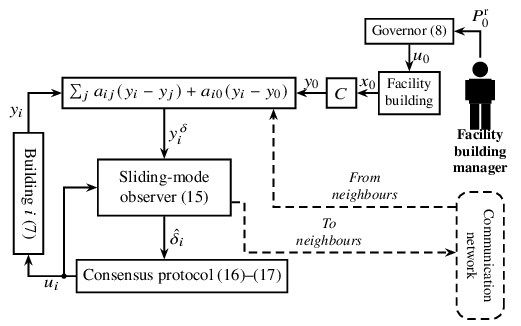} 
    \caption{The block diagram of the multi-agent control system}
    \label{fig:diag}
    \vspace{-5pt}
\end{figure}
\vspace{-0.2in}
\subsection{Observer Design}
\vspace{-0.05in}
This section is focused on designing distributed observers for the multi-building system to ensure robust state estimation in the presence of unknown input. Conventional linear designs of such observers require two conditions:
\begin{enumerate}
    \item $\text{rank}(C_iB_i)=\text{rank}(B_i)$.\\
    \item $\text{rank}
    \begin{bmatrix}
        sI-A_i & B_i\\
        C_i & 0
    \end{bmatrix}=n+m \; \forall s\in\mathbb{C}, \; \text{Re}(s) \ge 0$
\end{enumerate}

Condition 1 ensures that the number of unknown inputs impacting the output matches those influencing the system states. Condition 2 excludes invariant zeros in the closed right half-plane.

The challenge in this application is that condition 1, known as the observer-matching condition, does not hold. Therefore, the linear unknown-input observers of \cite{9062580, 280770, 1278} cannot be applied, while the distributed observer of \cite{HONG2008846} requires the leader's input to be known.

For the system in (\ref{eq:ss}) with $n=2$ and $m=1$, the Rosenbrock matrix is
\begin{equation}
P_i(s)=\begin{bmatrix} s & -1 & 0\\ 0 & s-a_i & b_i\\ 1 & 0 & 0
\end{bmatrix},\qquad \det P_i(s) = -b_i,
\label{eq:detcomp}
\end{equation}
which is independent of $s$ and nonzero, since
$b_i=-\eta_i/(2\Delta\theta_i C_i^{\mathrm{th}} n_i)\neq 0$ for any building. Hence \eqref{eq:detcomp} shows that $\mathrm{rank}\,P_i(s)=n+m=3$
for all $s\in\mathbb{C}$. The system has no invariant zeros and is
strongly observable \cite{HAUTUS1983353}, meaning $x_i$ can be reconstructed from $y_i$ and
its derivatives regardless of the unknown input acting on it. However,
$C_iB_i = \begin{bmatrix} 1 & 0 \end{bmatrix} \begin{bmatrix} 0\\ b_i \end{bmatrix} = 0$, which implies that $\text{rank}(C_iB_i) = 0$, not equal to  $\text{rank}(B_i)=1$. Therefore, the observer matching condition is violated. That is,  the unknown input does not enter $\dot\varepsilon_i = \xi_i$ directly; it acts only through $\dot\xi_i$. Hence, the unknown input appears at relative degree two from \(y_i^\delta\). This motivates a second-order, super-twisting sliding-mode observer.

Its discontinuous injection \cite{KALSI2010347} enforces a sliding manifold on the estimation error dynamics. Once the system reaches the sliding surface, the equivalent injection compensates for the effect of unknown inputs without requiring the matching condition. 
\vspace{-0.15in}
\subsection{Sliding-mode Observer Design}
\vspace{-0.05in}
Each agent $i$ has distinct system matrices $A_i$ and $B_i$ based on its thermal characteristics, but the same output matrix $C_i = C = \begin{bmatrix} 1 & 0 \end{bmatrix}$. 

Furthermore, the output of the consensus error can be measured as in (\ref{eq:10}):
\begin{equation}\label{eq:10}
y_i^{\delta} = \sum_{j \in \mathcal{N}_i} a_{ij}(y_i - y_j) + a_{i0}(y_i - y_0) = C_i\delta_i
\end{equation}

Then the sliding-mode observer \cite{LEVANT1998379} for agent $i$ is designed in \eqref{eq:observer}:
\begin{equation}\label{eq:observer}
\begin{aligned}
\dot{\hat{\delta}}_{i,1} &= \hat{\delta}_{i,2}+\kappa_{i,1}|e_{y,i}|^{1/2}\,\mathrm{sgn}(e_{y,i}),\\
\dot{\hat{\delta}}_{i,2} &= a_i\hat{\delta}_{i,2}+c_ib_iu_i+\kappa_{i,2}\,\mathrm{sgn}(e_{y,i}),
\end{aligned}
\end{equation}
where $\hat{\delta}_i \in \mathbb{R}^2$ is the consensus error estimate, $e_{y,i} = y_i^{\delta} - C\hat{\delta}_i$ is the output estimation error, $c_i=\sum_{j\in\mathcal{N}_i}a_{ij}+a_{i0}$, and $\kappa_{i,1},\kappa_{i,2}>0$ are the agent-specific super-twisting gains.



Finally, the sliding-mode observer-based leader-follower consensus protocol for the HOTL  multi-agent system can be written as \eqref{eq:sat}:
\begin{equation}
    u_i
    =
    \operatorname{sat}_{\bar u_i}
    \left(
    \hat v_i
    \right).
    \label{eq:sat}
\end{equation}
where $\hat v_i$ is defined in \eqref{eq:vhat} as follows:
\feq{\label{eq:vhat}
\hat v_i = -K_i\hat{\delta}_i
}
Under the conditions of Section \ref{seq:4}, the saturation in \eqref{eq:sat} remains inactive.
\vspace{-0.2in}
\section{Stability Analysis}\label{seq:4}
Let $[N]:=\{1,\ldots,N\}$ denote the follower index set. The observer error is denoted by $e_i^\delta=\delta_i-\hat\delta_i=\mathrm{col}(e_{i,1},e_{i,2})$, and $\phi_i$ denotes the lumped perturbation entering the observer error dynamics of agent $i$.
\begin{assumption}[Spanning tree]\label{ass:A1}
    The augmented graph (follower edges and leader pinning) contains a spanning tree rooted at the leader node $v_{0}$. Then $L+D_{0}$ is nonsingular.
\end{assumption}

\subsection{Observer Error}
Time is measured in hours in this section. From \eqref{eq:delta}, \eqref{eq:ss}, and \eqref{eq:observer}, the observer error of agent $i$ satisfies
\begin{equation}
\begin{aligned}
\dot e_{i,1}&=e_{i,2}-\kappa_{i,1}|e_{i,1}|^{1/2}\,\mathrm{sgn}(e_{i,1}),\\
\dot e_{i,2}&=-\kappa_{i,2}\,\mathrm{sgn}(e_{i,1})+\phi_i,
\end{aligned}
\label{eq:err}
\end{equation}
where
\begin{equation}
\begin{aligned}
\phi_i={}&a_ie_{i,2}+\sum_{j\in\mathcal{N}_i}a_{ij}\big[(a_i-a_j)\xi_j-b_ju_j\big]\\
&+a_{i0}\big[(a_i-a_0)\xi_0-b_0u_0\big]+\lambda_i\dot{\vartheta}
\end{aligned}
\label{eq:phi}
\end{equation}
and $\lambda_i=c_i\ell_i-\sum_{j\in\mathcal{N}_i}a_{ij}\ell_j-a_{i0}\ell_0$. The input $u_i$ of agent $i$ does not appear in $\phi_i$, whereas the inputs of its neighbors and of the leader do. The sign function is understood in the Filippov sense, with $\mathrm{sgn}(0)=[-1,1]$.

Suppose that $|\xi_j(t)|\le X_j$ for $j=0,\dots,N$ and $|u_j(t)|\le U_j$ for $j\in[N]$, where the envelopes $X_j$ and $U_j$ are given in Section~\ref{sec:bounds}. Since $|u_0|\le s_0$ and $|\dot{\vartheta}|\le W$, \eqref{eq:phi} gives $|\phi_i(t)|\le|a_i||e_{i,2}(t)|+d_i^{\mathrm{obs}}$, where
\begin{equation}
\begin{aligned}
d_i^{\mathrm{obs}}={}&\sum_{j\in\mathcal{N}_i}a_{ij}\big(|a_i-a_j|X_j+|b_j|U_j\big)\\
&+a_{i0}\big(|a_i-a_0|X_0+|b_0|s_0\big)+|\lambda_i|W.
\end{aligned}
\label{eq:dobs}
\end{equation}

\begin{lemma}[Finite-time observer]\label{lm:finite_obs}
Let $E_1,E_2,\Delta_i>0$, choose $\kappa_{i,1}=6\sqrt{\Delta_i}$ and $\kappa_{i,2}=2\Delta_i$, and define
\begin{equation}
\begin{gathered}
S_i=22E_1+\frac{6\sqrt{E_1}\,E_2}{\sqrt{\Delta_i}}+\frac{E_2^2}{\Delta_i},\qquad T_i=2\sqrt{\frac{S_i}{\Delta_i}},\\
\bar e_{i,1}=\frac{S_i}{13},\qquad \bar e_{i,2}=\sqrt{\frac{22\Delta_iS_i}{13}}.
\end{gathered}
\label{eq:sta}
\end{equation}
If $|e_{i,1}(0)|\le E_1$, $|e_{i,2}(0)|\le E_2$, and $|\phi_i(t)|<\Delta_i$ on $[0,\tau)$, then $|e_{i,1}(t)|\le\bar e_{i,1}$ and $|e_{i,2}(t)|\le\bar e_{i,2}$ on $[0,\tau)$, and $e^\delta_i(t)=0$ for $T_i\le t<\tau$.
\end{lemma}

\begin{proof}
Following \cite{MorenoOsorio2012}, let $V_i=\zeta_i^\top\Pi\,\zeta_i$, where $\Pi\succ0$ and
\[
\zeta_i=\mathrm{col}\big(|e_{i,1}|^{1/2}\mathrm{sgn}(e_{i,1}),\,e_{i,2}/\sqrt{\Delta_i}\big),\quad
\Pi=\begin{bmatrix}22&-3\\-3&1\end{bmatrix}.
\]
The initial bounds give $V_i(0)\le S_i$. Where $e_{i,1}\neq0$, write $\phi_i=\nu_i\Delta_i\,\mathrm{sgn}(e_{i,1})$ with $|\nu_i|<1$. Then \eqref{eq:err} gives
\[
\dot\zeta_i=\frac{\sqrt{\Delta_i}}{|\zeta_{i,1}|}\begin{bmatrix}-3&1/2\\-2+\nu_i&0\end{bmatrix}\zeta_i,\qquad
\dot V_i=-\frac{\sqrt{\Delta_i}}{|\zeta_{i,1}|}\,\zeta_i^\top Q(\nu_i)\,\zeta_i,
\]
where, for $|\nu|\le1$,
\[
Q(\nu)=\begin{bmatrix}120+6\nu&-18-\nu\\-18-\nu&3\end{bmatrix}\succeq\frac{\Pi}{10}+\frac{5}{2}\begin{bmatrix}1&0\\0&0\end{bmatrix}.
\]
Hence $\dot V_i\le-\sqrt{\Delta_i}\big(V_i/(10|\zeta_{i,1}|)+5|\zeta_{i,1}|/2\big)\le-\sqrt{\Delta_i}\sqrt{V_i}$, where the last step uses $x+y\ge2\sqrt{xy}$. If $e_{i,1}=0$ while $V_i>0$, then $e_{i,2}\neq0$, so $e_{i,1}$ crosses zero and such instants are isolated; since $V_i$ is continuous, the inequality holds across them. Therefore $\sqrt{V_i(t)}\le\sqrt{S_i}-\sqrt{\Delta_i}\,t/2$ until $V_i$ reaches zero, which happens no later than $T_i$, and the same inequality keeps $V_i$ at zero afterwards. Finally, $V_i\le S_i$ and $\Pi^{-1}=\frac{1}{13}\left[\begin{smallmatrix}1&3\\3&22\end{smallmatrix}\right]$ give $|e_{i,1}|=\zeta_{i,1}^2\le S_i/13$ and $e_{i,2}^2/\Delta_i\le22S_i/13$.
\end{proof}

\subsection{Closed-Loop Deviation Dynamics}
Let $H=L+D_0$ and
\begin{equation}
\begin{aligned}
A_d&=\mathrm{diag}(A_1,\dots,A_N),\quad B=\mathrm{diag}(B_1,\dots,B_N),\\
K&=\mathrm{diag}(K_1,\dots,K_N),\quad A_c=A_d-BK(H\otimes I_2),
\end{aligned}
\label{eq:Ac}
\end{equation}
where $\mathrm{diag}(\cdot)$ denotes a block-diagonal matrix. The gains $K_i$ are chosen so that $A_c$ is Hurwitz. For the graph used in Section~V, $H$ is lower triangular with diagonal entries $c_i$, so the eigenvalues of $A_c$ are those of $A_i-c_iB_iK_i$. The choice $c_ib_iK_{i1}=\omega^2$ and $a_i-c_ib_iK_{i2}=-2\omega$ places both at $-\omega$.

Let $\chi=\mathrm{col}(\tilde{x},\xi_0)$, $q=\mathrm{col}(K_1e^\delta_1,\dots,K_Ne^\delta_N)$, $\hat v=\mathrm{col}(\hat v_1,\dots,\hat v_N)$, and
\begin{equation}
\begin{aligned}
M&=\begin{bmatrix}A_c&R_a\\0&a_0\end{bmatrix},\quad
B_q=\begin{bmatrix}B\\0\end{bmatrix},\quad
C_u=\begin{bmatrix}-K(H\otimes I_2)&0\end{bmatrix},\\
R_a&=\mathrm{col}(0,\,a_1-a_0,\,\dots,\,0,\,a_N-a_0),\\
D_p&=\mathrm{col}(-\mathbf{1}_N\otimes B_0,\,b_0),\\
D_\vartheta&=\mathrm{col}(0,\,\ell_1-\ell_0,\,\dots,\,0,\,\ell_N-\ell_0,\,\ell_0),
\end{aligned}
\label{eq:mats}
\end{equation}
where $\mathbf{1}_N$ is the vector of $N$ ones. Since $u_0=\dot p$ and $\hat v_i=-K_i\delta_i+K_ie^\delta_i$, while no input saturates \eqref{eq:ss}, \eqref{eq:delta}, \eqref{eq:sat}, and \eqref{eq:vhat} give
\begin{equation}
\dot\chi=M\chi+D_p\dot p+D_\vartheta\dot{\vartheta}+B_qq,\qquad \hat v=C_u\chi+q.
\label{eq:chi}
\end{equation}
The change of variables $\psi=\chi-D_pp-D_\vartheta\vartheta$ removes the derivatives of $p$ and $\vartheta$:
\begin{equation}
\begin{aligned}
\dot\psi&=M\psi+MD_pp+MD_\vartheta\vartheta+B_qq,\\
\hat v&=C_u\psi+C_uD_pp+C_uD_\vartheta\vartheta+q.
\end{aligned}
\label{eq:psi}
\end{equation}
Since $a_0<0$, $M$ is Hurwitz whenever $A_c$ is.

Write $\psi=\mathrm{col}(\psi_1,\dots,\psi_N,\psi_0)$ with $\psi_i=\mathrm{col}(\psi_{i,1},\psi_{i,2})$. Because the initial states satisfy \eqref{eq:model4}, relation \eqref{eq:model4} holds for all $t$, and with \eqref{eq:Pstar} it gives
\begin{equation}
\begin{aligned}
\varepsilon_0&=\varepsilon^\star+\psi_0/a_0,\qquad \varepsilon_i=\varepsilon^\star+\psi_{i,1}+\psi_0/a_0,\\
P_i&=P_i^\star+\big[\psi_{i,2}-a_i\psi_{i,1}+(1-a_i/a_0)\psi_0\big]/b_i,\\
\xi_0&=\psi_0+b_0p+\ell_0\vartheta,\qquad \xi_i=\psi_{i,2}+\psi_0+\ell_i\vartheta,
\end{aligned}
\label{eq:rows}
\end{equation}
together with $P_0=P_0^\star+p$. Each comfort index and power is therefore its equilibrium value plus a fixed row vector times $\psi$.

\subsection{Componentwise Bounds}\label{sec:bounds}
For a row vector $g$, let $Y_g^\alpha$, $Y_g$, and $Y_{gM}$ solve
\begin{equation}
\begin{aligned}
(M+\alpha I)^\top Y_g^\alpha+Y_g^\alpha(M+\alpha I)&=-g^\top g,\\
M^\top Y_g+Y_gM&=-g^\top g,\\
M^\top Y_{gM}+Y_{gM}M&=-(gM)^\top gM,
\end{aligned}
\label{eq:lyap}
\end{equation}
where $\alpha>0$ and $M+\alpha I$ is Hurwitz. Let $f(t)=ge^{Mt}v$. The Cauchy--Schwarz inequality with the weight $e^{-\alpha t}$ gives $\int_0^\infty|f|\,dt\le\sqrt{v^\top Y_g^\alpha v/(2\alpha)}$, and $f(t)^2=f(0)^2+2\int_0^tf\dot f$ gives $\sup_{t\ge0}|f(t)|\le\Gamma_g(v)$, where
\[
\Gamma_g(v)=\sqrt{(gv)^2+2\sqrt{(v^\top Y_gv)(v^\top Y_{gM}v)}}.
\]
Define
\begin{equation}
\begin{aligned}
\gamma_{p,g}&=\sqrt{(MD_p)^\top Y_g^\alpha MD_p/(2\alpha)},\\
\gamma_{\vartheta,g}&=\sqrt{(MD_\vartheta)^\top Y_g^\alpha MD_\vartheta/(2\alpha)},\\
\sigma_g&=-gD_p,\qquad \bar q_j=|K_{j1}|\bar e_{j,1}+|K_{j2}|\bar e_{j,2},\\
\beta_g&=\Gamma_g(\psi(0))+\textstyle\sum_{j=1}^N\Gamma_g(B_{q,j})\,\bar q_jT_j,\\
Z_g&=\bar p\big(|gD_p|+\gamma_{p,g}\big)+\Theta\big(|gD_\vartheta|+\gamma_{\vartheta,g}\big)+\beta_g,
\end{aligned}
\label{eq:defs}
\end{equation}
where $B_{q,j}$ is the $j$th column of $B_q$.

\begin{lemma}[Output bounds]\label{lm:out_bound}
Suppose that on $[0,\tau)$ no input saturates, $-\bar p\le p\le0$, $|\vartheta|\le\Theta$, and, for every $j\in[N]$, $|q_j|\le\bar q_j$ with $q_j=0$ for $t\ge T_j$. Then for every row vector $g$ and $t\in[0,\tau)$,
\begin{equation}
\begin{aligned}
g\psi(t)&\ge-\tfrac{\bar p}{2}(\gamma_{p,g}+\sigma_g)-\Theta\gamma_{\vartheta,g}-\beta_g,\\
g\psi(t)&\le\tfrac{\bar p}{2}(\gamma_{p,g}-\sigma_g)+\Theta\gamma_{\vartheta,g}+\beta_g,
\end{aligned}
\label{eq:band}
\end{equation}
and $|g\chi(t)|\le Z_g$.
\end{lemma}

\begin{proof}
By \eqref{eq:psi}, $g\psi(t)$ is the sum of $ge^{Mt}\psi(0)$ and the convolutions of $p$, $\vartheta$, and $q_j$ with $k_p(s)=ge^{Ms}MD_p$, $k_\vartheta(s)=ge^{Ms}MD_\vartheta$, and $ge^{Ms}B_{q,j}$. The first term is bounded by $\Gamma_g(\psi(0))$ and the $q$ terms by $\sum_j\Gamma_g(B_{q,j})\bar q_jT_j$. The $\vartheta$ term is bounded by $\Theta\gamma_{\vartheta,g}$. For the $p$ term, write $p=-\bar p/2+\tilde p$ with $|\tilde p|\le\bar p/2$. Since $\int_0^\infty k_p=\sigma_g$, the $p$ term equals $-(\bar p/2)\sigma_g+(\bar p/2)\int_t^\infty k_p+\int_0^tk_p(t-s)\tilde p(s)\,ds$, and the last two terms are together bounded in magnitude by $(\bar p/2)\int_0^\infty|k_p|\le(\bar p/2)\gamma_{p,g}$. This gives \eqref{eq:band}. The bound on $g\chi$ follows from $\chi=\psi+D_pp+D_\vartheta\vartheta$, $|p|\le\bar p$, and $|g\psi|\le\bar p\gamma_{p,g}+\Theta\gamma_{\vartheta,g}+\beta_g$.
\end{proof}

The comfort and power bounds follow from \eqref{eq:band} with the rows in \eqref{eq:rows}. The envelopes in \eqref{eq:dobs} are $X_j=Z_g$ for the rows $g$ with $g\chi=\xi_j$, and $U_i=Z_g+\bar q_i$ with $g$ the $i$th row of $C_u$.

\subsection{Main Result}
\begin{theorem}\label{thm:main}
Let Assumption~1 hold, let the leader power be produced by the governor \eqref{eq:governor} with $s_0<\bar u_0$, and let the ambient temperature be absolutely continuous with $|\vartheta(t)|\le\Theta$ and $|\dot{\vartheta}(t)|\le W$ for all $t\ge0$. Let the measurements be exact, and let the initial states satisfy \eqref{eq:model4} with $|e_{i,1}(0)|\le E_1$ and $|e_{i,2}(0)|\le E_2$. If $M$ is Hurwitz, the observer gains are chosen as in Lemma~\ref{lm:finite_obs}, and
\begin{equation}
\Delta_i>d_i^{\mathrm{obs}}+|a_i|\bar e_{i,2},\qquad U_i<\bar u_i,\qquad i\in[N],
\label{eq:cond}
\end{equation}
then every Filippov solution of the closed loop satisfies, for all $t\ge0$ and $i\in[N]$:
\begin{enumerate}
\item $|e_{i,k}(t)|\le\bar e_{i,k}$ for $k=1,2$, and $e^\delta_i(t)=0$ for $t\ge T_i$;
\item $|\hat v_i(t)|\le U_i<\bar u_i$, so the saturation in \eqref{eq:sat} is never active;
\item every comfort index and power, including those of the leader, satisfies the bounds obtained from \eqref{eq:rows} and \eqref{eq:band}, and the tracking error $\chi(t)$ is bounded.
\end{enumerate}
\end{theorem}

\begin{proof}
Let $\Phi_i(t)$ be $|a_i||e_{i,2}(t)|$ plus the right-hand side of \eqref{eq:dobs} with $X_j$ and $U_j$ replaced by $|\xi_j(t)|$ and $|u_j(t)|$. Then $|\phi_i(t)|\le\Phi_i(t)$, and $\Phi_i$ and $\hat v_i$ are continuous. Let $\tau$ be the first time at which $\Phi_i(\tau)=\Delta_i$ or $|\hat v_i(\tau)|=\bar u_i$ for some $i$, with $\tau=\infty$ if no such time exists. Since $|e_{i,k}(0)|\le E_k<\bar e_{i,k}$, the argument below also applies at $t=0$ and gives $\Phi_i(0)<\Delta_i$ and $|\hat v_i(0)|<\bar u_i$, so $\tau>0$. On $[0,\tau)$, $|\phi_i|<\Delta_i$, so Lemma~\ref{lm:finite_obs} gives item~1 there. No input saturates on $[0,\tau)$, so Lemma~\ref{lm:out_bound} applies and gives $|\xi_j|\le X_j$, $|u_j|\le|\hat v_j|\le U_j$, and item~3. Hence $\Phi_i(t)\le d_i^{\mathrm{obs}}+|a_i|\bar e_{i,2}$ and $|\hat v_i(t)|\le U_i$ on $[0,\tau)$. If $\tau$ were finite, continuity and \eqref{eq:cond} would give $\Phi_i(\tau)<\Delta_i$ and $|\hat v_i(\tau)|<\bar u_i$ for all $i$, which contradicts the definition of $\tau$. Therefore $\tau=\infty$.
\end{proof}

Theorem~\ref{thm:main} gives bounded tracking while the leader power changes. Exact tracking follows once the leader power and the ambient temperature settle.

\begin{corollary}\label{cry:1}
Under the conditions of Theorem~\ref{thm:main}, if $p(t)\to p_\infty$ and $\vartheta(t)\to\vartheta_\infty$ as $t\to\infty$, then $x_i(t)-x_0(t)\to0$ for every follower and $\xi_0(t)\to0$.
\end{corollary}

\begin{proof}
For $t\ge\max_iT_i$, $q=0$, and \eqref{eq:psi} is an exponentially stable linear system with convergent inputs. Hence $\psi\to-D_pp_\infty-D_\vartheta\vartheta_\infty$, and $\chi=\psi+D_pp+D_\vartheta\vartheta\to0$.
\end{proof}

For the parameters and the equilibrium initial state of Section~V, with $E_1=10^{-6}$, $E_2=0.1$~h$^{-1}$, $\Delta_i=30$~h$^{-2}$, and $\alpha=0.5$~h$^{-1}$, the matrix $M$ is Hurwitz, $\bar e_{i,2}<0.154$~h$^{-1}$, $T_i<28.4$~s, and $\min_i(\Delta_i-d_i^{\mathrm{obs}}-|a_i|\bar e_{i,2})>13.0$~h$^{-2}$. The bounds of Theorem~\ref{thm:main} then give $0.4129<\varepsilon_i(t)<0.9621$ and $P_i(t)>44.90$~kW for all buildings, $\max_iP_i(t)/P_i^{\max}<0.640$, and $\max_i|\hat v_i(t)|/\bar u_i<0.602$, for every revision sequence admitted by the governor. These values are computed in floating-point arithmetic from \eqref{eq:lyap}. Theorem~\ref{thm:main} concerns the continuous-time closed loop with exact measurements; sampled-data implementation and measurement noise are not covered.
\vspace{-0.15in}
\section{Simulation and results}\label{seq:5}
This section presents simulation results demonstrating the performance and effectiveness of the proposed method. The simulation model includes a multi-building community, such as a university campus, consisting of six buildings in which TCLs serve as the primary controllable loads. {\color{blue}Also, the control input $u_i(t)$ is constrained by a maximum ramp rate  $\bar{u}_i = n_iP/T_{\mathrm{ramp}}$, with $T_{\mathrm{ramp}} = 300$~s to reflect finite TCL power-ramping capability}. Building B$_0$ serves as the facility management building, acting as the leader, which receives real-time DR signals from the utility. $B_0$ is connected to the five follower buildings $B_1$--$B_5$ through the communication graph shown in Fig.~3, over which the buildings exchange only their comfort states. The follower buildings, which may include classrooms, laboratories, and/or office spaces, are each characterized by distinct TCL parameters, comfort bands, and thermal capacitance as shown in Table \ref{tab:tab1}. The thermal capacitances are chosen so that the thermal time constants $R^{th}_iC^{th}_i$ lie between 0.5 and 1.25~h, allowing the comfort states to respond within the duration of a DR event. The model is simulated over 24~h with a fixed-step third-order solver at a step of 0.01~s, and the remaining design parameters are listed in Table \ref{tab:cont}.  The ambient temperature is held at $T^\star=30$~$^\circ$C, so $\vartheta=0$ and $w_i=0$. All buildings start at the equilibrium in \eqref{eq:Pstar}, with $\varepsilon_i(0)=0.5$, $\xi_i(0)=0$, $P_i(0)=P_i^\star$, and $\hat{\delta}_i(0)=0$. The manager converts a total reduction into a leader target by dividing it by 6.43, the steady-state ratio of the total reduction to the leader reduction obtained from Table~I. A 250~kW commitment thus gives $p=-38.9$~kW, inside the governor band $\bar{p}=40$~kW. 

Three SCE-inspired DR scenarios \footnote{See SCE's ELRP and CPP program pages: \href{https://www.sce.com/business/demand-response/emergency-load-reduction-program}{ELRP}, \href{https://www.sce.com/business/rates/cpp}{CPP}.} are considered, each with a DR event from 16:00 to 20:00: baseline DR (\hyperref[sc:sc1]{Scenario A}), where the manager fully approves the 250~kW request; mid-course revision (\hyperref[sc:sc2]{Scenario B}), where the target is reduced from 250~kW to 200~kW at 16:30, before the comfort states have settled; and late revision (\hyperref[sc:sc3]{Scenario C}), where the manager initially commits to 250~kW and later reduces it to 200~kW at 18:00, after the buildings have settled.
\begin{figure} 
    \centering
    \includegraphics[width=0.9\linewidth]{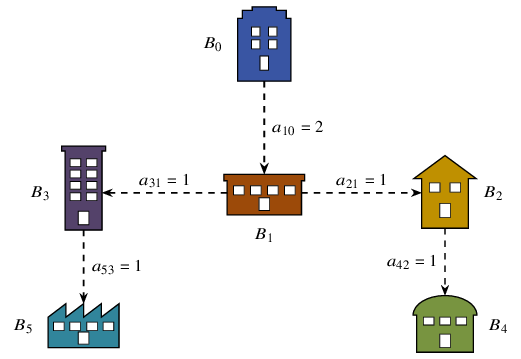} 
    \caption{Communication graph}
    \label{fig:com}
\end{figure}


\sisetup{
  round-mode          = places,
  round-precision     = 2,
  table-number-alignment = center,
  table-format        = 3.2 
}
\newcolumntype{Y}{>{\centering\arraybackslash}X}

\begin{table}[!t]
  \caption{Single-AC parameters for the six buildings}
  \label{tab:tab1}
  \centering
  \scriptsize                      
  \setlength{\tabcolsep}{3pt}      
  \renewcommand{\arraystretch}{1.1} 

  \begin{tabularx}{\columnwidth}{@{}l *{6}{Y}@{}}
    \toprule
    \textbf{Parameter} & \textbf{B0} & \textbf{B1} & \textbf{B2} & \textbf{B3} & \textbf{B4} & \textbf{B5} \\ \midrule
    Performance coefficient $\eta$ & 2.50 & 2.50 & 2.50 & 2.50 & 2.50 & 2.50 \\
    Temperature tolerance $\Delta\theta$ [\si{\celsius}] & 2 & 3 & 1 & 2 & 3 & 1 \\
    Temperature setpoint $\theta^{\text{set}}$ [\si{\celsius}] & 25 & 23 & 24 & 25 & 24 & 23 \\
    Thermal resistance $R^{\text{th}}$ [\si{\celsius\per\kilo\watt}] & 1.50 & 1.25 & 1.25 & 1.00 & 1.00 & 1.40 \\
    Thermal capacitance $C^{\text{th}}$ [\si{\kilo\watt\hour\per\celsius}] & 0.7 & 1 & 0.9 & 0.5 & 0.7 & 0.8 \\
    Number of {\color{blue}TCL loads} $n_i$ & 100 & 80 & 100 & 70 & 80 & 100 \\
    Rated power per unit $P$ [\si{\kilo\watt}] & 5 & 5 & 5 & 5 & 5 & 5 \\
    \bottomrule
  \end{tabularx}
  \vspace{-0.15in}

\end{table}

\begin{table}[t]
\centering
\caption{Design and simulation parameters}\label{tab:cont}
\footnotesize
\setlength{\tabcolsep}{4pt}
\renewcommand{\arraystretch}{1.15}
\begin{tabular}{@{}lrrrrr@{}}
\hline
 & $B_1$ & $B_2$ & $B_3$ & $B_4$ & $B_5$ \\
\hline
\multicolumn{6}{@{}l}{\textit{Controller}} \\
$K_{i1}$ [kW\,s$^{-1}$] & $-0.960$ & $-0.720$ & $-0.560$ & $-1.344$ & $-0.640$ \\
$K_{i2}$ [kW] & $-1075.2$ & $-800.0$ & $-560.0$ & $-1420.8$ & $-710.9$ \\
\hline
\multicolumn{6}{@{}l}{\textit{Luenberger observer (Section V-D)}} \\
$L_{i1}$ [$10^{-3}$\,s$^{-1}$] & 4.67 & 5.19 & 11.67 & 8.33 & 5.21 \\
$L_{i2}$ [$10^{-6}$\,s$^{-2}$] & 4.89 & 6.04 & 30.56 & 15.59 & 6.09 \\
\hline
\multicolumn{6}{@{}l}{\textit{Common parameters}} \\
$\kappa_{i1}$ [s$^{-1}$] & \multicolumn{2}{l}{$9.13\times10^{-3}$} & \multicolumn{2}{l}{$\bar{p}$ [kW]} & \multicolumn{1}{l@{}}{40} \\
$\kappa_{i2}$ [s$^{-2}$] & \multicolumn{2}{l}{$4.63\times10^{-6}$} & \multicolumn{2}{l}{$s_0$ [kW\,s$^{-1}$]} & \multicolumn{1}{l@{}}{0.25} \\
$T^\star$ [$^\circ$C] & \multicolumn{2}{l}{30} & \multicolumn{2}{l}{$h_g$ [s]} & \multicolumn{1}{l@{}}{1} \\
$\Theta$ [$^\circ$C] & \multicolumn{2}{l}{0.25} & \multicolumn{2}{l}{$W$ [$^\circ$C\,h$^{-1}$]} & \multicolumn{1}{l@{}}{1} \\
\hline
\end{tabular}
\end{table}


\newcommand{\panelwidth}{0.95\linewidth}
\newcommand{\panelheight}{5.5cm}           
\begin{figure}[htbp]
    \centering
        \includegraphics[height=\panelheight]{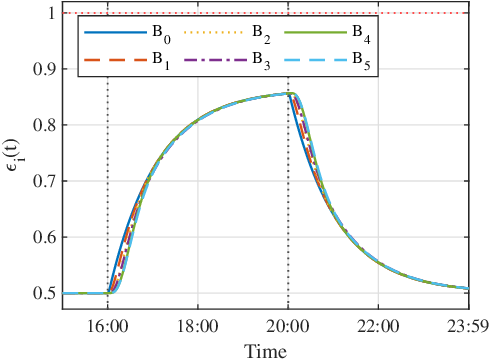}
    \hfill
    \caption{Comfort-state consensus among agents}\label{fig:baseline_comfort}
\end{figure}
\vspace{-0.15in}
\subsection{Scenario A: Baseline DR}\label{sc:sc1}
In this scenario, the facility manager receives a request from the utility to reduce demand by 250~kW from 16:00 to 20:00. The facility building, acting as the leader, then moves to a new comfort state. The other buildings then reach a consensus on the new comfort state by following the leader's state, as demonstrated in Fig.~\ref{fig:baseline_comfort}, and the spread in comfort index between buildings, which peaks at 0.060 at 16:19, falls below 0.001 by the end of the event. All comfort indices remain within $[0.5, 0.86]$, inside the acceptable range $[0,1]$. Fig.~\ref{fig:power_cmp} shows the load falling from 1037~kW to 787~kW, a reduction of 249~kW relative to the pre-event load. These values are inside the bounds of Theorem~\ref{thm:main}: the comfort indices span $[0.50, 0.86]$ against the certified $[0.413, 0.962]$, the lowest building power is 94.4~kW against 44.9~kW, and the largest power and input reach 48\% and 13\% of their limits against 64\% and 60\%. Fig.~\ref{fig:ob}  shows the estimation errors of the observer. The estimation error of the comfort-state consensus error remains below $2\times10^{-9}$ throughout, and that of its time derivative below $2\times10^{-7}$~s$^{-1}$. Both errors are exactly zero before 16:00, when all buildings are at equilibrium. Since $e^\delta_i(0)=0$ and $|\phi_i|<\Delta_i$, Lemma~\ref{lm:finite_obs} gives $e^\delta_i(t)=0$ for all $t$ in continuous time, the small values after 16:00 come from the discrete-time integration of the discontinuous injection. The unknown leader's command, therefore, does not disturb the estimates.

\begin{figure}[htbp]
    \centering

    \subcaptionbox{\label{fig:first_ob}}{%
        \includegraphics[height=\dimexpr\panelheight*65/100\relax]{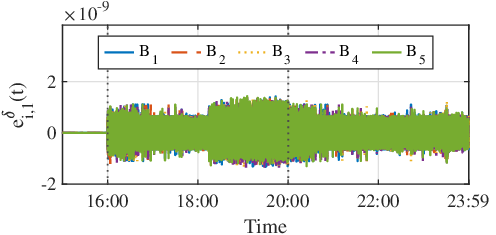}
    }


    \subcaptionbox{\label{fig:second_ob}}{%
        \includegraphics[height=\dimexpr\panelheight*65/100\relax]{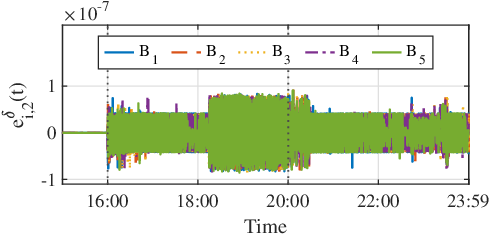}
    }

    \caption{Observer-error trajectories during the baseline DR event. (a) Estimation error of the comfort-state consensus error. (b) Estimation error of the time derivative of the comfort-state consensus error.}
   \label{fig:ob}
    \vspace{-0.2in}
\end{figure}

\vspace{-0.1in}
\subsection{Scenario B: Mid-course Revision\label{sc:sc2}}
\vspace{-2pt}
In Scenario B, the facility manager initially decides to shed 250~kW from the campus load. However, at 16:30, before the comfort states have settled, the manager revises the curtailment to 200~kW. As Fig.~\ref{fig:power_cmp} shows, at the moment of the revision the reduction is 254 kW, above the revised level, because the buildings are still responding to the initial 250 kW target. The reduction then falls to 191 kW at 17:09 and stays within 5 kW of the revised 200 kW target from 17:41. The corresponding indoor temperatures, shown in Fig.~\ref{fig:stemp}, rise modestly under the initial curtailment target while remaining within the acceptable comfort band, and then adapt to the lower target. Fig.~\ref{fig:scomf} confirms that all follower buildings quickly realign with the leader’s updated state, demonstrating robust consensus tracking control after the revision of the manager's command. 

\vspace{-5pt}
\begin{figure}[!t]
  \centering



  \subcaptionbox{\label{fig:stemp}}{%
    \includegraphics[height=\panelheight]{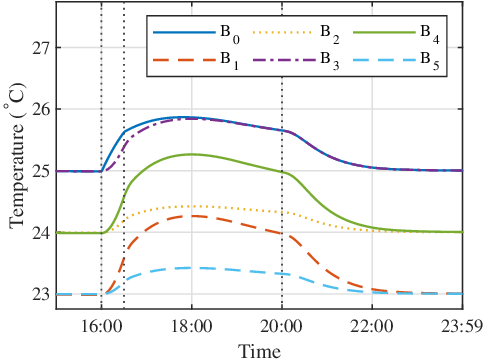}
  }

  \vspace{-2pt}

  \subcaptionbox{\label{fig:scomf}}{%
    \includegraphics[height=\panelheight]{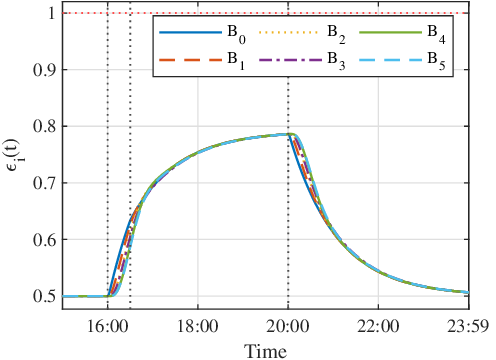}
  }

  \caption{Performance during the mid-course revision scenario. (a) Indoor temperature responses. (b) Occupant-comfort indices. Vertical dotted lines mark the start of the event, the revision, and the end of the event.}
  \label{fig:dr_sudden}
  \vspace{-0.25in}
\end{figure}
\vspace{-0.15in}
\subsection{Scenario C: Late Revision}\label{sc:sc3}
\vspace{-0.05in}
The facility manager initially commits to 250~kW, reducing campus load from about 1037 kW to about 790 kW, and allows the buildings to settle. At 18:00, the manager reduces the commitment to 200~kW, for example, after observing that several buildings are operating close to their limits because of an equipment fault. The reduction comes within 5 kW of the revised level by 18:16, dips to 196 kW at 18:33, and the load returns to its pre-event level after the event ends. Figs.~\ref{fig:ctl} and \ref{fig:ccl} show the indoor temperatures and comfort indices recovering from their peak of 0.82 at 18:15, while the event is still in progress. The traces also show close follower-leader alignment.

Fig.~\ref{fig:power_cmp} compares the total campus load in the three scenarios with the load the campus would draw without any DR action. The two revision scenarios settle at the same level but differ in their transient: in Scenario~B the reduction had already exceeded the revised commitment when the revision was issued, whereas in Scenario~C the load eases toward it from a settled operating point. With the ambient temperature held constant, the no-DR load stays at 1037 kW, and the reductions at 20:00 are 249 kW in Scenario A and 200 kW in Scenarios B and C.
\vspace{-2pt}
\begin{figure}[!t]
  \centering



  \subcaptionbox{\label{fig:ctl}}{%
    \includegraphics[height=\panelheight]{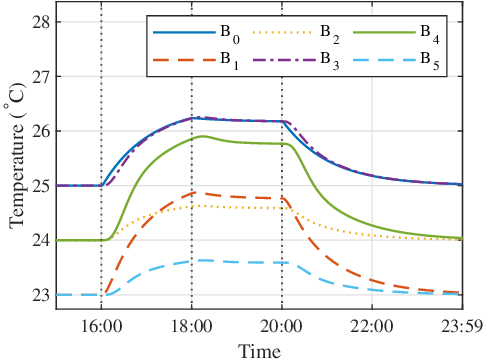}%
  }

  \vspace{-2pt}

  \subcaptionbox{\label{fig:ccl}}{%
    \includegraphics[height=\panelheight]{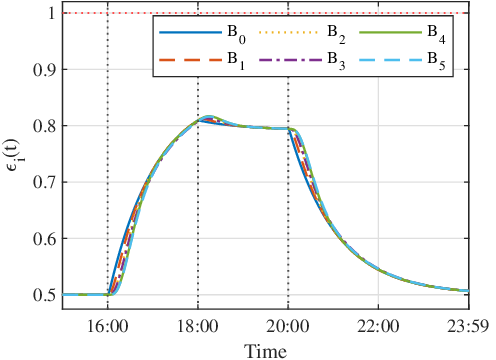}%
  }

  \caption{Performance during the late revision scenario. (a) Indoor temperatures. (b) Occupant-comfort indices. Vertical dotted lines mark the start of the event, the revision, and the end of the event.}
  \label{fig:dr_window}
  \vspace{-0.15in}
\end{figure}
\vspace{-0.2pt}

\begin{figure}[!t]
    \centering
    \includegraphics[height=\panelheight]{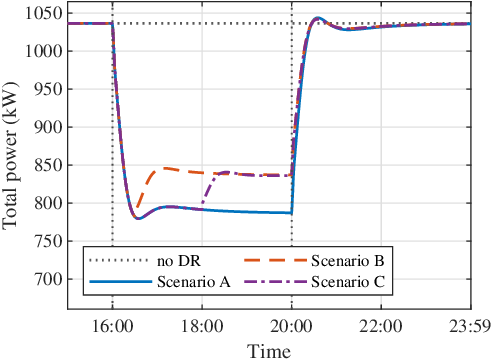}
    \hfill
    \caption{Power profile comparison between all scenarios}\label{fig:power_cmp}
    \vspace{-0.15in}
\end{figure}
\begin{figure}[!t]
  \centering


  \vspace{-2pt}

  \subcaptionbox{\label{fig:lu_smo_rms1}}{%
    \includegraphics[height=\panelheight]{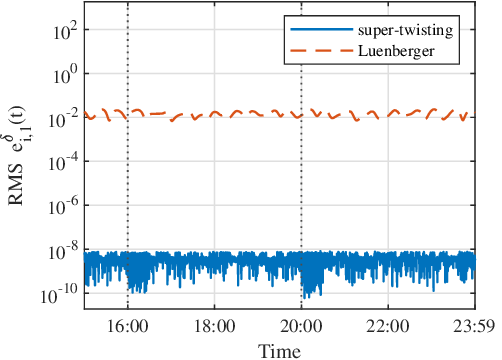}%
  }

  \vspace{-2pt}

  \subcaptionbox{\label{fig:lu_comf}}{%
    \includegraphics[height=\panelheight]{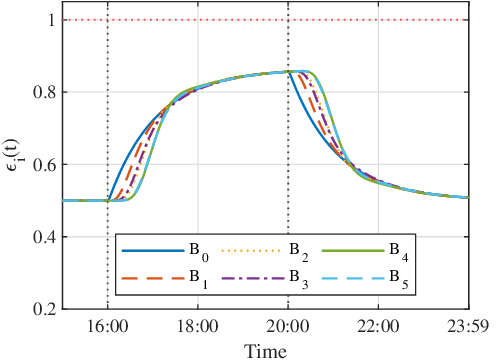}%
  }

  \vspace{-2pt}

 \caption{Comparison with a Luenberger observer placed at ten and twelve times each building's own thermal pole. (a) RMS over the five followers at each instant, $\sqrt{\tfrac{1}{N}\sum_{i} e_{i,1}^2(t)}$, of the estimation error of the comfort-state consensus error. (b) Comfort-state consensus under the Luenberger observer.}
  \label{fig:lu_comp}
  \vspace{-0.15in}
\end{figure}
\vspace{-0.2pt}
\subsection{Comparison with a Luenberger Observer}
As discussed in Section~III-C, the observer-matching condition fails for this system, so the linear unknown-input observers cannot be constructed. A conventional Luenberger observer, however, can be constructed whenever the pair $(A_i, C_i)$ is observable, and it therefore serves as a baseline. The Luenberger observer poles are placed at ten and twelve times each building's own thermal pole, with every other element of the system unchanged; its gains are listed in Table~\ref{tab:cont}. It uses the same own-input term $c_iB_iu_i$ as the observer in \eqref{eq:observer}, so only the correction term differs. The analysis of Section~\ref{seq:4} covers only the super-twisting observer.

As shown in Fig.~\ref{fig:lu_smo_rms1}, the RMS estimation error of the comfort-state consensus error, pooled over the five followers and the DR event (16:00--20:00, 1\,s samples) and computed identically for both observers, rises from $6.1\times10^{-10}$ to $3.5\times10^{-3}$, and that of its time derivative from $2.8\times10^{-8}$~s$^{-1}$ to $1.9\times10^{-5}$~s$^{-1}$, a factor of about 670. This degradation propagates to the control objective. Fig.~\ref{fig:lu_comf} shows that consensus is still reached, but more slowly: the spread in comfort index peaks at 0.137 instead of 0.060 and falls below 0.01 at 17:54 instead of 16:55. The reduction also overshoots to 289 kW instead of 257 kW, and the rebound after the event grows from 7 kW to 39 kW. The reason is that the super-twisting injection compensates $\phi_i$ exactly while $|\phi_i|<\Delta_i$, whereas the error of a linear observer is driven by $\phi_i$ and shrinks only as its bandwidth grows.

\section{Conclusion}\label{seq:6}
This paper introduced a HOTL approach for controlling multi-building systems participating in DR programs. Treating the facility manager's curtailment command as a bounded, unmeasured input allows the buildings to coordinate DR participation without access to that command. The state required by the consensus protocol cannot be measured, and because the command is unknown, estimating that state with a conventional linear observer requires the observer-matching condition, which does not hold for this system. By leveraging the strong observability of the system, a super-twisting sliding-mode observer is proposed that does not require this condition. The observer reconstructs the consensus error in finite time, and the tracking error, comfort indices, and powers stay within proven bounds, while the human leader, acting as the facility manager, balances user thermal comfort and grid support.

The simulation results show that the proposed observer and consensus protocol maintain coordination under supervisory commands that are not communicated to the participating buildings, while a Luenberger observer designed at a conventional bandwidth reaches consensus more slowly, with about twice the comfort spread and a larger rebound after the event. 
\vspace{-3pt}

%



\FloatBarrier
\vspace{-0.06in}
\section*{Acknowledgment}
\vspace{-4pt}
This work is supported in part by CEC grant EPC-16-077.
\vspace{-7pt}
\ifCLASSOPTIONcaptionsoff
  \newpage
\fi


\bibliographystyle{IEEEtran}
\bibliography{References}
\vspace{-0.25in}
%







\end{document}